\documentclass[preprint,authoryear,3p,12pt]{elsarticle}

\usepackage{amssymb}
\usepackage{amsthm}

\usepackage{mathrsfs}
\usepackage{float}
\usepackage{eufrak}

\journal{}

\newtheorem{theorem}{Theorem}[section]

\newtheorem{lemma}[theorem]{Lemma}

\newtheorem{define}[theorem]{Definition}

\def\k{{\rm K}}

\def\M{{\bf M}}

\def\lpp{{\rm lpp}}
\def\lc{{\rm lc}}
\def\f{{\bf f}}
\def\e{{\bf e}}

\def\u{{\bf u}}
\def\v{{\bf v}}
\def\w{{\bf w}}
\def\bru{\bar{\u}}

\def\brf{\bar{f}}

\def\lcm{{\rm lcm}}

\def\deg{{\rm deg}}
\def\max{{\rm max}}

\def\x{{x_1,\cdots,x_n}}

\def\lif{{\bf if \,}}
\def\lthen{{\bf then \,}}

\def\lendif{{\bf end if\,}}

\def\lwhile{{\bf while \,}}

\def\lendwhile{{\bf end while\,}}
\def\ldo{{\bf do \,}}

\def\lend{{\bf end \,}}
\def\lreturn{{\bf return \,}}
\def\lla{{\longleftarrow}}

\def\lbegin{{\bf begin}}

\newcommand{\SPC}{\hspace*{7pt}}
\newcommand{\comment}[1]{}
\newcommand{\ignore}[1]{}
\newcommand{\gr}{Gr\"obner\,}

\begin{document}

\begin{frontmatter}



\title{A Generalized Criterion for  Signature Related \gr Basis Algorithms\tnoteref{label0}}
\tnotetext[label0]{Version 1.2.}


\author{Yao Sun and Dingkang Wang\fnref{label1}}

\fntext[label1]{The authors are supported by NSFC 10971217 and 60821002/F02.}

\address{Key Laboratory of Mathematics Mechanization, Academy of Mathematics and Systems Science, CAS, Beijing 100190,  China}

\ead{sunyao@amss.ac.cn, dwang@mmrc.iss.ac.cn}

\begin{abstract}
A generalized criterion for signature related algorithms  to compute \gr basis is proposed  in this paper. Signature related algorithms are a popular kind of algorithms for computing \gr basis, including the famous F5 algorithm, the extended F5 algorithm and the GVW algorithm. The main purpose of current paper is to study in theory what kind of criteria is correct in   signature related algorithms and provide a generalized method to develop new criteria. For this purpose, a generalized criterion is proposed. The generalized criterion only relies on a general partial order defined on a set of  polynomials. When specializing the partial order to appropriate specific orders, the generalized criterion can specialize to almost all existing criteria of signature related algorithms. For {\em admissible} partial orders, a complete proof for the correctness of the algorithm based on this generalized criterion is also presented. This proof has no extra requirements on the computing order of critical pairs, and is also valid for non-homogeneous polynomial systems. More importantly, the partial orders implied by existing criteria are admissible. Besides, one can also check whether a new criterion is correct in signature related algorithms or even develop new criteria by using other admissible partial orders in the generalized criterion.
\end{abstract}


\end{frontmatter}



\section{Introduction}

\gr basis was first proposed by Buchberger in 1965 \citep{Buchberger65}. Since then, many important improvements have been made to speed up the algorithm for computing \gr basis \citep{Buchberger79,  Lazard83, Buchberger85, GebMol86, Gio91, Mora92, Fau99, Fau02}. Up to now,  F5  is one of the most efficient algorithms for computing \gr basis.
The concept of signatures for polynomials was also introduced   by Faug\`ere in \citep{Fau02}. Since F5 was proposed in 2002, it has been widely investigated and several   variants of F5  have   been presented, including the F5C algorithm \citep{Eder09} and F5 with extended criteria \citep{Ars09}.  Gao et al proposed an incremental algorithm G2V to compute \gr basis in \citep{Gao09}, and presented an extended version GVW in \citep{Gao10b}.

The common characteristics of all the above algorithms  are (1)  each polynomial has been assigned a  {\em signature}, and (2) both the criteria and the  reduction process  depend on the signatures of polynomials.  The only difference among the algorithms is that their criteria are different.

By studying  the criteria carefully, we find that all of these criteria work almost in a same way.
Suppose $f$ and $g$ are polynomials with signatures and the S-pair of $f$ and $g$ is denoted by $(t_f, f, t_g, g)$ where $t_f$ and $t_g$ are power products such that the leading power product of $t_f f$ and $t_g g$ are the same. Then a necessary condition of existing criteria to reject this S-pair is that, there exists some known polynomial $h$ such that $h$'s signature is a factor of $t_f f$'s
or $t_g g$'s signature. However, this condition is not sufficient to make the criteria correct.
Thus, existing criteria use different extra conditions to ensure correctness. With this insight,
we generalize  these extra conditions to a partial order defined on a set of polynomials, and
then propose a generalized criterion for signature related algorithms.
Therefore, when specializing the partial order to appropriate specific orders,
the generalized criterion can specialize to almost all existing criteria of signature
related algorithms. We emphasize that the generalized criterion can not only specialize
to a {\em single} criterion, but also can specialize to {\em several} criteria at the
same time. We will discuss the specializations in detail.


Unfortunately, not all general partial orders can make the generalized criterion correct. We proved   that the generalized criterion is correct if the partial order  is {\em admissible}. Unlike other proofs for the correctness of signature related algorithms \citep{Fau02, Stegers05, Eder08, Eder09, Ars09, Gao10b}, the proof in this paper is  complete. The proof does not need   extra requirements on  the computing order of critical pairs.  The proof is also not limited   to homogeneous polynomial systems. At present, most proofs for signature related algorithms always assume the input polynomial system is homogeneous or the critical pair with the smallest signature is computed first.
However, in practical implementation, these extra requirements usually make the algorithm less efficient. Moveover, we show  that the partial orders implied by the criteria of F5 and GVW are both admissible, so the proof in this paper is also valid for the correctness of F5 and GVW.  A complete proof for the correctness of F5 is also given in \citep{SunWang10a, SunWang10b}.

The significance of the generalized criterion is to show what kind of criteria for signature related algorithms is correct and provide a generalized method to check or develop new criteria. Specifically, when a new criterion is presented, if it can be specified from the generalized criterion by using an admissible partial order, then this new criterion is definitely correct. It is also possible for us to develop some new criteria by using an admissible partial order in the generalized criterion. From  the proof in this paper, we know that any admissible partial order can develop a new criterion for signature related algorithms in theory, but not all of these criteria are really efficient. Therefore, we claim that if the admissible partial order is in fact a total order, then almost all useless computations can be avoided. The proof for the claim will be included in our future works.

The paper is organized as follows. Section 2 gives the generalized criterion and describes how this generalized criterion specializes to the criteria of F5 and GVW. Section 3 proves the correctness of the generalized criterion. Section 4 discusses a new criterion by using an admissible partial order, and conducts some comparisons. Concluding remarks follow in Section 6.

\section{Generalized Criterion} \label{sec_criterion}

\subsection{Generalized criterion}

Let $R=\k[\x]$ be a polynomial  ring over a field $\k$ with $n$ variables. Suppose $\{f_1, \cdots, f_m\}$  is a finite subset of  $ R$.  We want to compute a \gr basis for the ideal
$$I=\langle f_1, \cdots, f_m\rangle=\{p_1f_1+\cdots+p_mf_m \mid p_1,\cdots,p_m \in R\} $$
with respect to some term order on $R$.


Let $\f=(f_1, \cdots, f_m) \in R^m$, and consider the following $R$-module of $R^m\times R$: $$\M=\{(\u, f) \in R^m\times R \mid \u\cdot\f=f\}.$$ Let $\e_i$ be the $i$-th unit vector of $R^m$, i.e. $(\e_i)_j=\sigma_{ij}$. Then the $R$-module $\M$ is generated by $\{(\e_1, f_1), \cdots, (\e_m, f_m)\}.$

Fix {\em any} term order $\prec_1$ on $R$ and {\em any} term $\prec_2$ on $R^m$. We must emphasize that the order $\prec_2$ may or may not be related to $\prec_1$ in theory, although $\prec_2$ is usually an extension of $\prec_1$ to $R^m$ in implementation. For sake of convenience, we shall use the following convention for leading power products: $$\lpp(f)=\lpp_{\prec_1}(f) \mbox{ and } \lpp(\u)=\lpp_{\prec_2}(\u),$$ for any $f\in R$ and $\u\in R^m$. We make the convention that if $f=0$ then $\lpp(f)=0$ and $\lpp(f) \prec_1 t$ for any non-zero power product $t$ in $R$; similarly for $\lpp(\u)$. In the following, we use $\prec$ to represent $\prec_1$ and $\prec_2$, if no confusion occurs.

For any $(\u, f)\in \M$, we call $\lpp(\u)$ the {\bf signature} of $(\u, f)$, which is the same as the signature used in F5.

Given a finite set $B\subset \M$, consider a {\bf partial order} ``$\leq$" defined on $B$, where ``$\leq$" has:
\begin{enumerate}

\item Reflexivity: $(\u, f)\leq (\u, f)$ for all $(\u, f)\in B$.

\item Antisymmetry: $(\u, f)\leq (\v, g)$ and $(\v, g) \leq (\u, f)$ imply $(\u, f)=(\v, g)$, where $(\u, f), (\v, g)\in B$.

\item Transitivity: $(\u, f)\leq (\v, g)$ and $(\v, g)\leq (\w, h)$ imply $(\u, f)\leq (\w, h)$, where $(\u, f)$, $(\v, g)$, $(\w, h) \in B$.

\end{enumerate}
In the rest of this paper, we {\em do not} care about the {\em equality} case, so we always use ``$<$", which means ``$\leq$" without equality.

Based on a partial order, we   give a generalized criterion for signature related algorithms.
\begin{define}[generalized rewritable criterion]
Given a set $B\subset \M$ and a partial order ``$<$" defined on $B$. We say $t(\u, f)$, where $(\u, f)\in B$, $f$ is nonzero and $t$ is a power product in $R$, is {\bf generalized rewritable} by $B$ ({\bf gen-rewritable} for short), if there exists $(\u', f')\in B$ such that
\begin{enumerate}

\item $\lpp(\u')$ divides $\lpp(t\u)$, and

\item $(\u', f') < (\u, f)$.

\end{enumerate}
\end{define}

In subsection \ref{subsec_specializations}, we will show how the generalized criterion specializes to some exiting criteria. In next subsection, we   describe how this generalized criterion is applied.

\subsection{Algorithm with generalized criterion} \label{subsec_algorithm}

Let $$G=\{(\v_1, g_1), \cdots, (\v_s, g_s)\}\subset \M $$ be a finite subset. We call $G$  an {\bf S-Gr\"obner basis} for $\M$ (``S" short for signature related), if for any $(\u, f)\in \M$, there exists $(\v, g)\in G$ such that
\begin{enumerate}

\item $\lpp(g)$ divides $\lpp(f)$, and

\item $\lpp(t\v) \preceq \lpp(\u)$, where $t=\lpp(f)/\lpp(g)$.

\end{enumerate}
If $G$ is an S-\gr basis for $\M$, then the set $\{g\mid (\v, g) \in G\}$ is a \gr basis of the ideal $I=\langle f_1, \cdots, f_m\rangle$. The reason is that for any $f\in \langle f_1, \cdots, f_m\rangle$, there exist $p_1, \cdots, p_m \in R$ such that $f=p_1f_1 + \cdots + p_m f_m$. Let $\u=(p_1, \cdots, p_m)$. Then $(\u, f)\in \M$ and hence there exists $(\v, g)\in G$ such that $\lpp(g)$ divides $\lpp(f)$ by the definition of S-\gr basis.

Suppose $(\u, f), (\v, g)\in \M$ are two pairs with $f$ and $g$ both nonzero. Let $t=\lcm(\lpp(f), \lpp(g))$, $t_f=t/\lpp(f)$ and $t_g=t/\lpp(g)$. If $\lpp(t_f\u) \succeq \lpp(t_g\v)$, then $$[t_f(\u, f), t_g(\v, g)]$$ is called a {\bf critical pair} of $(\u, f)$ and $(\v, g)$. The corresponding {\bf S-polynomial} is $t_f(\u, f)-ct_g(\v, g)$ where $c=\lc(f)/\lc(g)$. Please keep in mind that, for any critical pair $[t_f(\u, f), t_g(\v, g)]$, we always have $\lpp(t_f\u) \succeq \lpp(t_g\v)$. Also notice that $t_f$ (or $t_g$) here does not mean it only depends on $f$ (or $g$). For convenience, we say $[t_f(\u, f), t_g(\v, g)]$ is a critical pair of $B$, if both $(\u, f)$ and $(\v, g)$ are in $B$.

Given a critical pair $[t_f(\u, f), t_g(\v, g)]$, there are three possible cases, assuming $c=\lc(f)/\lc(g)$:

\begin{enumerate}

\item If $\lpp(t_f\u - ct_g\v) \not= \lpp(t_f\u)$, then we say $[t_f(\u, f), t_g(\v, g)]$ is {\bf non-regular}.

\item If $\lpp(t_f\u - ct_g\v) = \lpp(t_f\u)$ and $\lpp(t_f\u) = \lpp(t_g\v)$, then $[t_f(\u, f), t_g(\v, g)]$ is called {\bf super regular}.

\item If $\lpp(t_f\u) \succ \lpp(t_g\v)$, then we call $[t_f(\u, f), t_g(\v, g)]$ {\bf genuine regular} or {\bf regular} for short.

\end{enumerate}


 We say a {\bf  critical pair  $[t_f(\u, f), t_g(\v, g)]$ is  gen-rewritable  } if {\em either} $t_f(\u, f)$ {\em or} $t_g(\v, g)$ is gen-rewritable.

We can now state the signature related \gr basis algorithm.

\smallskip
\noindent {\bf GB algorithm with generalized criterion (GBGC)}\\
{\bf Input: } $(\e_1, f_1),\cdots, (\e_m, f_m)$ \\
{\bf Output: } An S-\gr basis for $M=\langle (\e_1, f_1),\cdots (\e_m, f_m) \rangle $  \\
\lbegin\\
  \SPC $G \lla \{(\e_i, f_i)\mid i=1, \cdots, m\}$\\
  \SPC $\mbox{\em CPairs} \lla \{[t_f(\u, f), t_g(\v, g)] \mid (\u, f), (\v, g)\in G\}$\\
  \SPC $G \lla G \cup \{(f_j\e_i - f_i\e_j, 0) \mid 1 \leq i < j \leq m\}$  \SPC \SPC \SPC \SPC \SPC ($\divideontimes$)  \\
  \SPC \lwhile {\em CPairs} $\not= \emptyset$  \ldo\\
  \SPC\SPC  $[t_f(\u, f), t_g(\v, g)] \lla $ {\bf any} critical pair  in \mbox{\em CPairs} \SPC $(\bigstar)$ \\
  \SPC\SPC  $\mbox{\em CPairs} \lla \mbox{\em CPairs} \setminus \{[t_f(\u, f), t_g(\v, g)]\}$\\
  \SPC\SPC  \lif    $[t_f(\u, f), t_g(\v, g)]$ is {\bf not gen-rewritable} by $G$   \\
    \SPC\SPC\SPC \,   and $[t_f(\u, f), t_g(\v, g)]$ is {\bf regular} \SPC \SPC \SPC \SPC  \SPC \SPC \SPC   \,  \, ($\divideontimes$) \\
  \SPC\SPC\SPC \lthen\\
  \SPC\SPC\SPC\SPC $c\lla \lc(f)/\lc(g)$\\
  \SPC\SPC\SPC\SPC $(\w, h) \lla$ reduce $t_f(\u, f)-c t_g(\v, g)$ by $G$ \\
  \SPC\SPC\SPC\SPC \lif $h\not=0$,\\
  \SPC\SPC\SPC\SPC\SPC \lthen \\
  \SPC\SPC\SPC\SPC\SPC\SPC $\mbox{\em CPairs} \lla \mbox{\em CPairs} \cup$ $\{$critical pair of \\
  \SPC\SPC\SPC\SPC\SPC\SPC\SPC\SPC $(\w, h) \mbox{ and } (\w', h') \mid (\w', h')\in G \mbox{ and } h'\not=0 \}$\\
    \SPC\SPC\SPC\SPC  \SPC \SPC   $G \lla G \cup \{(h\e_i - f_i\w, 0) \mid i =1,\cdots,m \}$   \,  \, \, ($\divideontimes$)  \\
  \SPC\SPC\SPC\SPC \lendif \SPC\SPC\SPC\SPC\SPC\SPC\SPC\SPC  \\
  \SPC\SPC\SPC\SPC  $G \lla G \cup \{(\w, h)\}$\\
  \SPC\SPC \lendif\\
  \SPC \lendwhile\\
  \SPC \lreturn $G$\\
\lend

For the above algorithm, please notice that
 \begin{enumerate}
\item The gen-rewritable criterion uses a partial order defined on $G$. While new elements are added to $G$, the partial order on $G$ needs to be updated simultaneously. Fortunately, most partial orders can be updated automatically.
\item  For the line ended with ($\bigstar$), we emphasize that any critical pair can be selected, while some other algorithm, such as GVW, always select the critical pair with minimal signature.
\item  The algorithm GBGC is still correct even without the lines  ended with ($\divideontimes$), but the algorithm  will do some redundant computations, and hence become less efficient.

\item For sake of efficiency, it suffices to record $(\lpp(\u), f)$ for each $(\u, f)\in G$ in the practical implementation.
\end{enumerate}

Next let us see the reduction process in the above algorithm. Given $(\u, f)\in \M$ and $B\subset \M$,
  $(\u, f)$ is said to be {\bf  reducible} by $B$, if there exists $(\v, g)\in B$ such that $g\not=0$, $\lpp(g)$ divides $\lpp(f)$ and $\lpp(\u-c t \v)=\lpp(\u)$ where $c=\lc(f)/\lc(g)$ and $t=\lpp(f)/\lpp(g)$. If $(\u, f)$ is reducible by some $(\v, g) \in B$, we say $(\u, f)$ {\bf reduces} to $(\u, f) - c t (\v, g) = (\u - c t \v, f-c t g)$ by $(\v, g)$ where $c=\lc(f)/\lc(g)$ and $t=\lpp(f)/\lpp(g)$. This procedure is called a one-step reduction. Next, we can repeat this process  until it is  not reducible by $B$ anymore.

There are some other  ways to define the reduction process   \citep{Gao10b, Ars09, Fau02}  and all of them have a common point.  That is  $\lpp(\u) = \lpp(\u - c t \v)$, which is a key characteristic of signature related algorithms.



In the GBGC algorithm, we say a partial order ``$<$" defined on $G$ is {\bf admissible}, if for any critical pair $[t_f(\u, f), t_g(\v, g)]$,  which   is regular and not gen-rewritable by $G$ when it is being selected from {\em CPairs} and whose corresponding S-polynomial is reduced to $(\w, h)$ by $G$, we always have $(\w, h)<(\u, f)$ after updating ``$<$" for $G \cup \{(\w, h)\}$.
We emphasize that in the above definition of admissible, the relation $(\w, h)<(\u, f)$ is
essential and $(\w, h)$ may not be related to other elements in $G$.

With the above definition, it is easy to verify whether  a partial order is admissible  in an algorithm. In next subsection, we will show  that the partial orders implied by the criteria in   F5 and GVW  are admissible.

If the algorithm GBGC terminates in finite steps, then we have the following theorem.
\begin{theorem}  \label{thm_main}
Let $\M=\langle (\e_1,f_1),\cdots, (\e_m,f_m) \rangle$ be an $R$-module in $R^m \times R$. Then an S-\gr basis for $M$ can be constructed by the algorithm GBGC if the partial order in the generalized criterion is admissible.
\end{theorem}

\subsection{Specializations} \label{subsec_specializations}

In this subsection, we focus on specializing the generalized criterion to the criteria of F5 and GVW   by using appropriate admissible partial orders in the algorithm GBGC.


\subsubsection{Criteria of F5}

First, we  list the criteria in F5 by current notations. In F5, the order $\prec_2$ on $R^m$ is obtained by extending $\prec_1$ to $R^m$ in a POT fashion with $\e_1\succ_2 \cdots \succ_2 \e_m$.

\begin{define}[syzygy criterion]
Given a set $B\subset \M$, we say $t(\u, f)$, where $(\u, f)\in B$ with $\lpp(\u) = x^\alpha \e_i$, $f$ is nonzero and $t$ is a power product in $R$, is {\bf F5-divisible} by $B$, if there exists $(\u', f')\in B$ with $\lpp(\u')=x^\beta \e_j$,  such that
\begin{enumerate}

\item $\lpp(f')$ divides $tx^\alpha$, and

\item $\e_i \succ \e_j$.

\end{enumerate}
\end{define}

\begin{define}[rewritten criterion]
Given a set $B\subset \M$, we say $t(\u, f)$, where $(\u, f)\in B$ and $t$ is a power product in $R$, is {\bf F5-rewritable} by $B$, if there exists $(\u', f')\in B$ such that
\begin{enumerate}

\item $\lpp(\u')$ divides $\lpp(t\u)$, and

\item $(\u', f')$ is added to $B$ later than $(\u, f)$.

\end{enumerate}
\end{define}

In F5, given a critical pair $[t_f(\u, f), t_g(\v, g)]$ of $B$, if either $t_f(\u, f)$ or $t_g(\v, g)$ is F5-divisible or F5-rewritable by $B$, then this critical pair is redundant.

Next, we show how to specialize the generalized criterion to both syzygy criterion and rewritten criterion at the same time. For this purpose, we choose the following partial order defined on $G$ which can be updated automatically when a new element is added to $G$: we say $(\u', f')<(\u, f)$ where $(\u', f'), (\u, f)\in G$, if
\begin{enumerate}

\item $f'=0$ and $f\not=0$,

\item otherwise, $(\u', f')$ is added to $G$ later than $(\u, f)$.

\end{enumerate}
The above partial order ``$<$" is admissible in the algorithm GBGC. Because for any critical pair $[t_f(\u, f), t_g(\v, g)]$,   which is regular and not gen-rewritable by $G$ when it is being selected from {\em CPairs} and whose corresponding S-polynomial is reduced to $(\w, h)$ by $G$, the pair $(\w, h)$ is always added to $G$ later than $(\u, f)$ no matter   $h $ is $0$ or not.

At last, we show how the generalized criterion specializes to the rewritten criterion and
syzygy criterion. For the rewritten criterion, the specialization is obvious by the definition
of ``$<$". For the syzygy criterion, if $t(\u, f)$, where $(\u, f)\in G $
with $\lpp(\u) = x^\alpha \e_i$ and $f \not= 0$, is F5-divisible by
some $(\u', f')\in G$ with $\lpp(\u')=x^\beta \e_j$, we have $\lpp(f')$
divides $tx^\alpha$ and $\e_i \succ \e_j$. According to the algorithm GBGC,
since $f'\not = 0$, we have $(f'\e_i - f_i\u', 0)\in G$
and $\lpp(f'\e_i - f_i\u') = \lpp(f')\e_i$ divides $t x^\alpha \e_i$. So $t(\u, f)$
is gen-rewritable by $(f'\e_i - f_i\u', 0)\in G$ by definition.


With a similar discussion, the generalized criterion can also specialize to the criteria in \citep{Ars09}, since the extended F5 algorithm in that paper only differs from the original F5 in the order $\prec_2$ on $R^m$.

\subsubsection{Criteria of GVW}

 First, we rewrite the criteria  in GVW by current notations.

\begin{define}[First Criterion]
Given a set $B\subset \M$. We say $t(\u, f)$, where $(\u, f)\in B$, $f$ is nonzero and $t$ is a power product in $R$, is {\bf GVW-divisible} by $B$, if there exists $(\u', f')\in B$ such that
\begin{enumerate}

\item $\lpp(\u')$ divides $\lpp(t\u)$, and

\item $f'=0$.

\end{enumerate}
\end{define}

\begin{define}[Second Criterion]
Given a set $B\subset \M$. We say $t(\u, f)$, where $(\u, f)\in B$ and $t$ is a power product in $R$, is {\bf eventually super top-reducible} by $B$, if $t(\u, f)$ is reducible and reduced to $(\w, h)$ by $B$, and then there exists $(\u', f')\in B$ such that
\begin{enumerate}

\item $\lpp(\u')$ divides $\lpp(\w)$, and

\item $\lpp(f')$ divides $\lpp(h)$, $\frac{\lpp(\w)}{\lpp(\u')} = \frac{\lpp(h)}{\lpp(f')}$ and $\frac{\lc(\w)}{\lc(\u')} = \frac{\lc(h)}{\lc(f')}$.

\end{enumerate}
\end{define}
In GVW, given a critical pair $[t_f(\u, f), t_g(\v, g)]$ of $B$, if $t_f(\u, f)$ is GVW-divisible or eventually super top-reducible by $B$, then this critical pair is redundant.
The GVW algorithm also has a third criterion.

\smallskip
\noindent{\bf Third Criterion}
{\em If there are two critical pairs $[t_f(\u, f), t_g(\v, g)]$ and $[\bar{t}_f(\bar{\u}, \bar{f}), \bar{t}_g(\bar{\v}, \bar{g})]$ of $B$ such that $\lpp(t_f\u) = \lpp(\bar{t}_f\bar{\u})$, then  at least one of the critical pairs is redundant.}

Next, in order to specialize the generalized criterion to the above three criteria at the same time, we use the following partial order defined on $G$ which can also be updated automatically when a new element is added to $G$: we say $(\u', f')<(\u, f)$ where $(\u', f'), (\u, f)\in G$, if one of the following two conditions holds:
\begin{enumerate}

\item  $\lpp(t'f') < \lpp(tf)$, where $t'= \frac{\lcm(\lpp(\u), \lpp(\u'))}{\lpp(\u')}$ and $t= \frac{\lcm(\lpp(\u), \lpp(\u'))}{\lpp(\u)}$ such that $\lpp(t'\u')=\lpp(t\u)$.

\item $\lpp(t'f') = \lpp(tf)$ and  $(\u', f')$ is added to $G$ later than $(\u, f)$.

\end{enumerate}
The above partial order ``$<$" is admissible in the algorithm GBGC. Because for any critical pair $[t_f(\u, f), t_g(\v, g)]$,  which  is regular and not gen-rewritable by $G$ when it is being selected from {\em CPairs} and whose corresponding S-polynomial is reduced to $(\w, h)$ by $G$, we always have $\lpp(t_f\u) = \lpp(\w)$ and $\lpp(t_ff) > \lpp(h)$.


At last, let us see the three criteria of GVW.

For the first criterion, if $t(\u, f)$ is GVW-divisible by some $(\u', f')\in G$, then $t(\u, f)$ is also gen-rewritable by $(\u', f') \in G$ by definition.

For the second criterion, if $t(\u, f)$, where $(\u, f)\in G$, is eventually super top-reducible by $G$, then $t(\u, f)$ is reduced to $(\w, h)$ and there exists $(\u', f')\in G$ such that $\lpp(\u')$ divides $\lpp(\w)$, $\lpp(f')$ divides $\lpp(h)$, $\frac{\lpp(\w)}{\lpp(\u')} = \frac{\lpp(h)}{\lpp(f')}$ and $\frac{\lc(\w)}{\lc(\u')} = \frac{\lc(h)}{\lc(f')}$. Then we have $\lpp(t'\u') = \lpp(\w) = \lpp(t\u)$ and $\lpp(t'f') = \lpp(h) < \lpp(tf)$, which means $(\u', f')<(\u, f)$. So $t(\u, f)$ is gen-rewritable by $(\u', f')\in G$.

For the third criterion, we have $\lpp(t_f\u) = \lpp(\bar{t}_f\bar{\u})$. First, if $(\u, f) < (\bar{\u}, \bar{f})$, then $\bar{t}_f(\bar{\u}, \bar{f})$ is gen-rewritable by $(\u, f)$ and hence $[\bar{t}_f(\bar{\u}, \bar{f}), \bar{t}_g(\bar{\v}, \bar{g})]$ is redundant; the reverse is also true. Second, if $(\u, f) = (\bar{\u}, \bar{f})$, one of the two critical pairs should be selected earlier from $CPairs$, assuming $[t_f(\u, f),$ $t_g(\v, g)]$ is selected first. If $[t_f(\u, f), t_g(\v, g)]$ is regular and not gen-rewritable, then its S-polynomial is reduced to $(\w, h)$ and $(\w, h)$ is added to $G$ by the algorithm GBGC. Since ``$<$" is admissible, we have $(\w, h) < (\u, f)$. Thus, when $[\bar{t}_f(\bar{\u}, \bar{f}), \bar{t}_g(\bar{\v}, \bar{g})]$ is selected afterwards, it will be redundant, since $\bar{t}_f(\bar{\u}, \bar{f})$ is gen-rewritable by $(\w, h)$. Otherwise, if $[t_f(\u, f), t_g(\v, g)]$ is not regular, or it is regular and gen-rewritable, then $[t_f(\u, f),$ $t_g(\v, g)]$ is redundant. Anyway, at least one of the critical pairs is redundant in the algorithm.


\section{Proofs for the Correctness of the Generalized Criterion} \label{sec_proof}

To prove the main theorem (Theorem \ref{thm_main}) of the paper, we  need the following definition and lemmas.

In this section, we always assume that $\M$ is an $R$-module generated by $\{(\e_1,f_1), \cdots, (\e_m,f_m)\}$. Let $(\u, f)\in \M$, we say $(\u, f)$ has a {\bf standard representation} w.r.t. a set $B\subset \M$, if there exist $p_1, \cdots, p_s \in R$ such that $$(\u, f) = p_1(\v_1, g_1)+\cdots+p_s(\v_s, g_s),$$ where $(\v_i, g_i)\in B$, $\lpp(\u) \succeq \lpp(p_i\v_i)$ and $\lpp(f)\succeq \lpp(p_ig_i)$ for $i=1,\cdots, s$. Clearly, if $(\u, f)$ has a  standard representation w.r.t. $B$, then there exists $(\v, g)\in B$ such that $\lpp(g)$ divides $\lpp(f)$ and $\lpp(\u) \succeq \lpp(t\v)$ where $t=\lpp(f)/\lpp(g)$.

\begin{lemma} \label{lem_correctness}
Let $G$ be a finite set of generators for $\M$. Then $G$ is an S-Gr\"obner basis for $\M$ if for any critical pair $[t_f(\u, f)$, $t_g(\v, g)]$ of $G$, the S-polynomial of $[t_f(\u, f), t_g(\v, g)]$   always has a standard representation w.r.t. $G$.
\end{lemma}

\begin{proof}
The proof of this lemma is direct by the theory of $t$-representation. For more details, please see \citep{Becker93}.
\end{proof}

\begin{lemma} \label{lem_stdrepresentation}
Let $G$ be a finite subset of $\M$ and $\{(\e_1, f_1),$ $\cdots, (\e_m, f_m)\}\subset G$. For an element $(\u, f)$ in $\M$, $(\u, f)$ has a standard representation w.r.t. $G$  if for any critical pair $[t_g(\v, g), t_h(\w, h)]$ of $G$ with $\lpp(\u) \succeq \lpp(t_g \v)$, the S-polynomial of $[t_g(\v, g), t_h(\w, h)]$  always has a standard representation w.r.t. $G$.
\end{lemma}

\begin{proof}
For $(\u, f)\in \M$, we have $\u\cdot \f = f$ where $\f = (f_1, \cdots, f_m)\in R^m$. Assume $\u=p_1\e_1+\cdots+p_m\e_m$  where $p_i\in R$. Clearly, $(\u, f) = p_1(\e_1, f_1) + \cdots + p_m(\e_m, f_m).$ Notice that $\lpp(\u)\succeq \lpp(p_i\e_i)$ for $i=1,\cdots,m$. If $\lpp(f)\succeq \lpp(p_if_i)$, then we have already got a standard representation for $(\u, f)$ w.r.t. $G$. Otherwise, we will prove it by the classical method. Let $T=\max\{\lpp(p_if_i)\mid i=1,\cdots,m\}$, then   $T\succ \lpp(f)$ holds by assumption.
Consider the equation
$$(\u, f)= \sum_{\lpp(p_if_i)=T} \lc(p_i)\lpp(p_i)(\e_i, f_i) +\sum_{\lpp(p_jf_j)\prec T}p_j(\e_j, f_j)$$ $$+ \sum_{\lpp(p_if_i)=T} (p_i - \lc(p_i)\lpp(p_i))(\e_i, f_i).\eqno(1)$$
The leading power products in the first sum should be canceled, since we have $T\succ \lpp(f)$. So the first sum can be rewritten as a sum of S-polynomials, that is $$\sum_{\lpp(p_if_i)=T}\lc(p_i)\lpp(p_i)(\e_i, f_i)= \sum \bar{c}t(t_g(\v, g)-ct_h(\w, h)),$$
 where $(\v, g), (\w, h)\in G$, $\bar{c}\in \k$, $t_g(\v, g)-ct_h(\w, h)$ is the S-polynomial of $[t_g(\v, g), t_h(\w, h)]$, $\lpp(t\ t_g g)=\lpp(t\ t_h h)=T$ and $\lpp(\u)\succeq \lpp(t\ t_g \v)\succeq \lpp(t\ t_h\w)$ such that we have   $\lpp(t ( t_g g - c  t_h h)) \prec T $. By the hypothesis of the lemma, the S-polynomial $(t_g\v-c t_h \w, t_g g-c t_h h)$ has a standard representation w.r.t. $G$, that is, $(t_g\v-c t_h \w, t_g g-c t_h h) = \sum q_i(\v_i, g_i)$, where $(\v_i, g_i)\in G$,  $\lpp(\u)\succeq \lpp(t\ t_g\v)\succeq \lpp(t\ q_i \v_i)$ and $\lpp(t_g g-c t_h h)\succeq \lpp(q_i g_i)$. Substituting these standard representations back to the original expression of $(\u, f)$ in $(1)$, we get a new representation for $(\u, f)$. Let $T^{(1)}$ be the maximal leading power product of the polynomial parts appearing in the right side of the new representation. Then we have $T\succ T^{(1)}$.  Repeat the above process until  $T^{(s)}$ is same as $\lpp(f)$ for some $s$ after finite steps. Finally, we always get a standard representation for $(\u, f)$.
\end{proof}

Before giving a full proof of the theorem, we introduce the following definitions first.

Suppose $[t_f(\u, f), t_g(\v, g)]$ and $[t_{f'}(\u', f'), t_{g'}(\v', g')]$ are two critical pairs, we say $[t_{f'}(\u', f'), t_{g'}(\v', g')]$ is {\bf smaller} than $[t_f(\u, f)$, $t_g(\v, g)]$  if one of the following conditions holds:
\begin{enumerate}

\item[(a).] $\lpp(t_{f'}\u') \prec \lpp(t_f\u)$.

\item[(b).] $\lpp(t_{f'}\u') = \lpp(t_f\u)$ and $(\u', f')<(\u, f)$.

\item[(c).] $\lpp(t_{f'}\u') = \lpp(t_f\u)$, $(\u', f') = (\u, f)$ and $\lpp(t_{g'}\v') \prec \lpp(t_g\v)$.

\item[(d).] $\lpp(t_{f'}\u') = \lpp(t_f\u)$,  $(\u', f') = (\u, f)$, $\lpp(t_{g'}\v') = \lpp(t_g\v)$ and $(\v', g')<(\v, g)$.

\end{enumerate}

 Let $D$ be a set of critical pairs. A critical pair     in $D$ is said to be  {\bf minimal}  if there is no critical pair in $D$  smaller  than this critical pair.
 The minimal critical pair in $D$ may not be unique, but we can always find one   if $D$ is finite.

Now, we can give the proof of the main theorem.

\begin{proof}[Proof of Theorem \ref{thm_main}] If the algorithm terminates in finite steps,  then  $G_{end}$ denotes the set returned by the algorithm GBGC.
Since $\{(\e_1, f_1), \cdots, (\e_m, f_m)\}\subset G_{end}$, then $G_{end}$ is a set of generators for $\M$. In the rest of this proof, we focus on showing $G_{end}$ is an S-\gr basis for $\M$.

We will take the following strategy to prove the theorem.\\
{\bf Step 1:}  Let $Todo$ be the set of  {\em all} the critical pairs of $G_{end}$, and $Done$ be an empty set.\\
 {\bf Step 2:} Select a minimal critical pair $[t_f(\u, f), t_g(\v, g)]$ in $Todo$. \\
  {\bf Step 3:} For such $[t_f(\u, f), t_g(\v, g)]$,  we will prove the following facts.
\begin{enumerate}
\item[(F1).] The  S-polynomial of  $[t_f(\u, f), t_g(\v, g)]$    has a standard representation w.r.t. $G_{end}$.

\item[(F2).] If $[t_f(\u, f), t_g(\v, g)]$ is {\em super regular} or {\em regular}, then $t_f(\u, f)$ is gen-rewritable by $G_{end}$.

\end{enumerate}
{\bf Step 4:} Move $[t_f(\u, f), t_g(\v, g)]$ from $Todo$ to $Done$, i.e. $Todo \lla Todo \setminus \{ [t_f(\u, f), t_g(\v, g)]\}$ and $Done \lla Done\ \cup$\\ $\{ [t_f(\u, f), t_g(\v, g)]\}$. \\
We can repeat {\bf Step 2, 3, 4} until $Todo$ is empty. Please notice that for every critical pair in $Done$, it always has property (F1). Particularly, if this critical pair is super regular or regular, then it has properties  (F1) and (F2). When $Todo$ is empty, all the critical pairs of $G_{end}$ will lie in $Done$, and hence, all the corresponding S-polynomials  have standard representations w.r.t. $G_{end}$. Then $G_{end}$ is an S-Gr\"obner basis by Lemma \ref{lem_correctness}.

{\bf Step 1, 2, 4} are trivial, so we next focus on showing the facts in {\bf Step 3}.

Take a minimal critical pair $[t_f(\u, f), t_g(\v, g)]$ in $Todo$. And this critical pair must appear in the algorithm GBGC. Suppose such pair is selected from the set $CPairs$ in  some  loop of the algorithm GBGC and $G_k$ denotes the set $G$ at the beginning of the same loop.  For such $[t_f(\u, f),$ $t_g(\v, g)]$, it must be in one of the following cases:
\begin{enumerate}

\item[C1:] $[t_f(\u, f), t_g(\v, g)]$ is {\em non-regular}.

\item[C2:] $[t_f(\u, f), t_g(\v, g)]$ is {\em super regular}.

\item[C3:] $[t_f(\u, f), t_g(\v, g)]$ is {\em regular} and is {\em not} gen-rewritable by $G_k$.

\item[C4:] $[t_f(\u, f), t_g(\v, g)]$ is {\em regular} and $t_f(\u, f)$ is gen-rewritable by $G_k$.

\item[C5:] $[t_f(\u, f), t_g(\v, g)]$ is {\em regular} and $t_g(\v, g)$ is gen-rewritable by $G_k$.

\end{enumerate}
Thus, to show the facts in {\bf Step 3}, it suffices to show (F1) holds in case {\bf C1},  and  (F1), (F2) hold in cases {\bf C2, C3, C4} and {\bf C5}. We will proceed for each case respectively.

 We  make the following claims under the condition that  $[t_f(\u, f), t_g(\v, g)]$ is  minimal in $Todo$. The proofs of these claims will be presented later.

{\bf Claim 1}: Given $(\bru, \brf)\in \M$, if $\lpp(\bru) \prec \lpp(t_f\u)$, then $(\bru, \brf)$ has a standard representation w.r.t. $G_{end}$.

{\bf Claim 2}: If $[t_f(\u, f), t_g(\v, g)]$ is super regular or regular and $t_f(\u, f)$ is gen-rewritable by $G_{end}$, then the S-polynomial of $[t_f(\u, f), t_g(\v, g)]$ has a standard representation w.r.t. $G_{end}$.

{\bf Claim 3:} If $[t_f(\u, f), t_g(\v, g)]$ is regular and $t_g(\v, g)$ is gen-rewritable by  $G_{end}$, then $t_f(\u, f)$ is also gen-rewritable by $G_{end}$.

Therefore, using {\bf Claim 2}, to show (F1) and (F2) hold in the cases {\bf C2, C3, C4} and {\bf C5}, it suffices to show $t_f(\u, f)$ is gen-rewritable by $G_{end}$ in each case.

{\bf C1:} $[t_f(\u, f), t_g(\v, g)]$ is {\em non-regular}. Consider the S-polynomial $(t_f\u-c t_g\v, t_f f-c t_g g)$ where $c=\lc(f)/\lc(g)$. Notice that $\lpp(t_f\u-c t_g\v) \prec \lpp(t_f\u)$ by the definition of non-regular, so {\bf Claim 1} shows $(t_f\u-c t_g\v, t_f f-c t_g g)$ has a standard representation w.r.t. $G_{end}$, which proves (F1).

{\bf C2:} $[t_f(\u, f), t_g(\v, g)]$ is {\em super regular}, i.e. $\lpp(t_f\u - ct_g\v) = \lpp(t_f\u)$ and $\lpp(t_f\u) = \lpp(t_g\v)$ where $c=\lc(f)/\lc(g)$. Let $\bar{c}=\lc(\u)/\lc(\v)$. Notice that $\bar{c}\not=c$, since $\lpp(t_f\u-c t_g \v)=\lpp(t_f\u)$. Then we have $\lpp(t_f \u-\bar{c} t_g \v)\prec \lpp(t_f\u)$ and $\lpp(t_f f-\bar{c} t_g g)=\lpp(t_f f)$. So {\bf Claim 1} shows $(t_f\u-\bar{c} t_g \v, t_f f-\bar{c} t_g g)$ has a standard representation w.r.t. $G_{end}$, and hence, there exists $(\w, h)\in G_{end}$ such that $\lpp(h)$ divides $\lpp(t_f f - \bar{c} t_g g)=\lpp(t_f f)$ and $\lpp(t_f \u)\succ \lpp(t_f \u-\bar{c} t_g \v) \succeq \lpp(t_h \w)$ where $t_h= \lpp(t_f f)/\lpp(h)$. Consider the critical pair of $(\u, f)$ and $(\w, h)$, say $[\bar{t}_f(\u, f), \bar{t}_h(\w, h)]$. Since $\lpp(h)$ divides $\lpp(t_f f)$, then $\bar{t}_f$ divides $t_f$, $\bar{t}_h$ divides $t_h$ and $\frac{\lpp(t_f)}{\lpp(\bar{t}_f)} = \frac{\lpp(t_h)}{\lpp(\bar{t}_h)}$. So $[\bar{t}_f(\u, f), \bar{t}_h(\w, h)]$ is regular and smaller than $[t_f(\u, f)$, $t_g(\v, g)]$ in fashion (a) or (b), which means $[\bar{t}_f(\u, f), \bar{t}_h(\w, h)]$ lies in $Done$ and $\bar{t}_f(\u, f)$ is gen-rewritable by $G_{end}$. Then ${t}_f(\u, f)$ is also gen-rewritable by $G_{end}$, since $\bar{t}_f$ divides $t_f$.

{\bf C3:} $[t_f(\u, f), t_g(\v, g)]$ is {\em regular} and {\em not} gen-rewritable by $G_k$. According to the algorithm GBGC, the S-polynomial $t_f(\u, f)-c t_g(\v, g)$ is reduced to $(\w, h)$ by $G_k$ where $c=\lc(f)/\lc(g)$, and $(\w, h)$ will be added to the set $G_k$ afterwards. Notice that $G_k\subset G_{end}$ and $(\w, h)\in G_{end}$. Since ``$<$" is an admissible partial order, we have $(\w, h)<(\u, f)$ by definition. Combined with the fact $\lpp(\w) = \lpp(t_f\u)$, so ${t}_f(\u, f)$ is gen-rewritable by $(\w, h)\in G_{end}$.

{\bf C4}: $[t_f(\u, f), t_g(\v, g)]$ is {\em regular} and $t_f(\u, f)$ is gen-rewritable by $G_k$. Then $t_f(\u, f)$ is also gen-rewritable by $G_{end}$, since $G_k\subset G_{end}$.

{\bf C5}: $[t_f(\u, f), t_g(\v, g)]$ is {\em regular} and $t_g(\v, g)$ is gen-rewritable by $G_k$.  $t_g(\v, g)$ is also gen-rewritable by $G_{end}$, since $G_k \subset G_{end}$. Then {\bf Claim 3} shows $t_f(\u, f)$ is gen-rewritable by $G_{end}$ as well.

After all, the theorem is proved.
\end{proof}

We give the proofs for the three claims below.

\begin{proof}[Proof of {\bf Claim 1}]
According to the hypothesis, we have $(\bru, \brf)\in \M$ and $\lpp(\bru)\prec \lpp(t_f\u)$. So for any critical pair $[t_{f'}(\u', f'), t_{g'}(\v', g')]$ of $G_{end}$ with $\lpp(\bru) \succeq \lpp(t_{f'} \u')$, we have $[t_{f'}(\u', f'), t_{g'}(\v', g')]$ is smaller than $[t_f(\u, f), t_g(\v, g)]$ in fashion (a) and hence lies in $Done$, which means the S-polynomial of $[t_{f'}(\u', f'), t_{g'}(\v', g')]$ has a standard representation w.r.t. $G_{end}$. So Lemma \ref{lem_stdrepresentation} shows that $(\bru, \brf)$ has a standard representation w.r.t. $G_{end}$.
\end{proof}

\begin{proof}[Proof of {\bf Claim 2}]
 We have that  $[t_f(\u, f), t_g(\v, g)]$ is minimal in $Todo$ and $t_f(\u, f)$ is gen-rewritable by $G_{end}$. Let $c = \lc(f)/\lc(g)$. Then $(\bru, \brf) = (t_f\u - c t_g \v, t_f f - c t_g g)$ is the S-polynomial of $[t_f(\u, f), t_g(\v, g)]$. Since $[t_f(\u, f), t_g(\v, g)]$ is super regular or regular, we have $\lpp(\bru) = \lpp(t_f \u)$. Next we will show that $(\bru, \brf)$ has a standard representation w.r.t. $G_{end}$. The proof is organized in the following way. \smallskip\\
{\bf First:} We show that there exists $(\u_0, f_0)\in G_{end}$ such that $t_f(\u, f)$ is gen-rewritable by $(\u_0, f_0)$ and $t_0(\u_0, f_0)$ is {\em not} gen-rewritable by $G_{end}$ where $t_0 = \lpp(t_f\u)/\lpp(\u_0)$.\smallskip\\
{\bf Second:} For such $(\u_0, f_0)$, we show that $\lpp(\brf) \succeq \lpp(t_0f_0)$ where $t_0 = \lpp(t_f\u)/\lpp(\u_0)$.\smallskip\\
{\bf Third:} We prove that $(\bru, \brf)$ has a standard representation w.r.t. $G_{end}$.\smallskip

Proof of the {\bf First} fact. By hypothesis, suppose $t_f(\u, f)$ is gen-rewritable by some $(\u_1, f_1)\in G_{end}$, i.e. $\lpp(\u_1)$ divides $\lpp(t_f\u)$ and $(\u_1, f_1) < (\u, f)$. Let $t_1 = \lpp(t_f \u) /\lpp(\u_1)$. If $t_1(\u_1, f_1)$ is not gen-rewritable by $G_{end}$, then $(\u_1, f_1)$ is the one we are looking for. Otherwise, there exists $(\u_2, f_2)\in G_{end}$ such that $t_1(\u_1, f_1)$ is gen-rewritable by $(\u_2, f_2)$. Notice that $t_f(\u, f)$ is also gen-rewritable by $(\u_2, f_2)$ and we have $(\u, f) > (\u_1, f_1) > (\u_2, f_2)$. Let $t_2 = \lpp(t_f \u) /\lpp(\u_2)$. We next discuss whether $t_2(\u_2, f_2)$ is gen-rewritable by $G_{end}$. In the better case, $(\u_2, f_2)$ is the needed one if $t_2(\u_2, f_2)$ is not gen-rewritable by $G_{end}$; while in the worse case, $t_2(\u_2, f_2)$ is gen-rewritable by some $(\u_3, f_3)\in G_{end}$. We can repeat the above discussions for the worse case. Finally, we will get a chain $(\u, f) > (\u_1, f_1) > (\u_2, f_2) > \cdots$. This chain must terminate, since $G_{end}$ is finite and ``>" is a partial order defined on $G_{end}$. Suppose $(\u_s, f_s)$ is the last one in the above chain. Then $t_f(\u, f)$ is gen-rewritable by $(\u_s, f_s)$ and $t_s(\u_s, f_s)$ is not gen-rewritable by $G_{end}$ where $t_s = \lpp(t_f\u)/\lpp(\u_s)$.

Proof of the {\bf Second} fact. From the {\bf First} fact, we have that $t_0(\u_0, f_0)$ is {\em not} gen-rewritable by $G_{end}$ where $t_0 = \lpp(t_f\u)/\lpp(\u_0)$. Next, we prove the {\bf Second} fact by contradiction. Assume $\lpp(\brf) \prec \lpp(t_0f_0)$. Let $c_0 = \lc(\bru)/\lc(\u_0)$. Then we have $\lpp(\bru - c_0 t_0 \u_0) \prec \lpp(\bru) = \lpp(t_0\u_0)$ and $\lpp(\brf - c_0 t_0 f_0) = \lpp(t_0 f_0)$. So $(\bru - c_0 t_0 \u_0, \brf - c_0 t_0 f_0)$ has a standard representation w.r.t. $G_{end}$ by {\bf Claim 1}, and hence, there exists $(\w, h)\in G_{end}$ such that $\lpp(h)$ divides $\lpp(\brf - c_0 t_0 f_0) = \lpp(t_0 f_0)$ and $\lpp(t_0\u_0) \succ \lpp(\bru - c_0 t_0 \u_0) \succeq \lpp(t_h\w)$ where $t_h=\lpp(t_0f_0)/\lpp(h)$. Next consider the critical pair of $(\u_0, f_0)$ and $(\w, h)$, say $[\bar{t}_0(\u_0, f_0)$, $\bar{t}_h(\w, h)]$. Since $\lpp(h)$ divides $\lpp(t_0f_0)$, then $\bar{t}_0$ divides $t_0$, $\bar{t}_h$ divides $t_h$ and $\frac{\lpp(t_0)}{\lpp(\bar{t}_0)} = \frac{\lpp(t_h)}{\lpp(\bar{t}_h)}$. So $[\bar{t}_0(\u_0, f_0)$, $\bar{t}_h(\w, h)]$ is regular and smaller than $[t_f(\u, f), t_g(\v, g)]$ in fashion (a) or (b), which means $[\bar{t}_0(\u_0, f_0)$, $\bar{t}_h(\w, h)]$ lies in $Done$ and $\bar{t}_0(\u_0, f_0)$ is gen-rewritable by $G_{end}$. Moreover, since $\bar{t}_0$ divides $t_0$,  $t_0(\u_0, f_0)$ is also gen-rewritable by $G_{end}$, which contradicts with the property that $t_0(\u_0, f_0)$ is {\em not} gen-rewritable by $G_{end}$. The {\bf Second} fact is proved.

Proof of the {\bf Third} fact. According to the second fact, we have $\lpp(\brf) \succeq \lpp(t_0f_0)$ where $t_0 = \lpp(t_f\u)/\lpp(\u_0)$. Let $c_0 = \lc(\bru)/\lc(\u_0)$. We have $\lpp(\bru - c_0 t_0 \u_0)\prec \lpp(\bru)$ and $\lpp(\brf - c_0 t_0 f _0)\preceq \lpp(\brf)$. So $(\bru, \brf) - c_0 t_0 (\u_0, f_0) = (\bru - c_0 t_0 \u_0, \brf - c_0 t_0 f _0)$ has a standard representation w.r.t. $G_{end}$ by {\bf Claim 1}. Notice that $\lpp(\bru)=\lpp(t_0 \u_0)$ and $\lpp(\brf)\succeq \lpp(t_0 f _0)$. So after adding $c_0 t_0 (\u_0, f_0)$ to both sides of the standard representation of $(\bru, \brf) - c_0 t_0 (\u_0, f_0)$, then we will get a standard representation of $(\bru, \brf)$ w.r.t. $G_{end}$.

{\bf Claim 2} is proved.
\end{proof}

\begin{proof}[Proof of  {\bf Claim 3}]
Since $t_g(\v, g)$ is gen-rewritable by $G_{end}$ and $\lpp(t_g\v) \prec \lpp(t_f\u)$, by using a similar method in the proof of {\bf Claim 2}, we can first show that there exists $(\v_0, g_0)\in G_{end}$ such that $t_g(\v, g)$ is gen-rewritable by $(\v_0, g_0)$ and $t_0(\v_0, g_0)$ is not gen-rewritable by $G_{end}$ where $t_0 = \lpp(t_g\v)/\lpp(\v_0)$. And then we can also prove that $\lpp(t_g g) \succeq \lpp(t_0 g_0)$ by contradiction.

If $\lpp(t_g g) = \lpp(t_0 g_0)$, then the critical pair of $(\u, f)$ and $(\v_0, g_0)$, say $[\bar{t}_f(\u, f), \bar{t}_0(\v_0, g_0)]$, must be regular and smaller than the critical pair $[t_f(\u, f), t_g(\v, g)]$ in fashion (a) or (d), which means $[\bar{t}_f(\u, f), \bar{t}_0(\v_0, g_0)]$ lies in $Done$ and $\bar{t}_f(\u, f)$ is gen-rewritable by $G_{end}$. Since $\lpp(t_0 g_0) = \lpp(t_g g) = \lpp(t_f f)$, then $\bar{t}_f$ divides $t_f$, and hence, $t_f(\u, f)$ is gen-rewritable by $G_{end}$ as well.

Otherwise, $\lpp(t_g g) \succ \lpp(t_0 g_0)$ holds. Let $c = \lc(\v)/\lc(\v_0)$, we have $\lpp(t_g \v - c t_0 \v_0) \prec \lpp(t_g \v)$ and $\lpp(t_g g - c t_0 g_0) = \lpp(t_g g)$. Then $(t_g \v - c t_0 \v_0, t_g g - c t_0 g_0)$ has a standard representation w.r.t. $G_{end}$ by {\bf Claim 1}, and hence, there exists $(\w, h)\in G_{end}$ such that $\lpp(h)$ divides $\lpp(t_g g - c t_0 g_0)=\lpp(t_g g)$ and $\lpp(t_h \w) \preceq \lpp(t_g \v - c t_0 \v_0) \prec \lpp(t_g \v)$ where $t_h = \lpp(t_g g)/\lpp(h)$. Then the critical pair of $(\u, f)$ and $(\w, h)$, say $[\bar{t}_f(\u, f), \bar{t}_h(\w, h)]$, must be regular and smaller than the critical pair $[t_f(\u, f), t_g(\v, g)]$ in fashion (a) or (c), which means $[\bar{t}_f(\u, f), \bar{t}_h(\w, h)]$ lies in $Done$ and $\bar{t}_f(\u, f)$ is gen-rewritable by $G_{end}$. Since $\lpp(h)$ divides $\lpp(t_g g) = \lpp(t_f f)$, then $\bar{t}_f$ divides $t_f$, and hence, $t_f(\u, f)$ is gen-rewritable by $G_{end}$ as well.

{\bf Claim 3} is proved.
\end{proof}

\section{New Criteria and comparisons}

Based on the generalized criterion, to develop new criteria for signature related algorithms, it suffices to choose appropriate admissible partial orders. For example, we can develop a new criterion by using the following admissible partial order implied by GVW's criteria: that is, $(\u', f')<(\u, f)$, where $(\u, f), (\u', f')\in G$, if one of the following two conditions holds.
\begin{enumerate}

\item $\lpp(t'f') < \lpp(tf)$ where $t'= \frac{\lcm(\lpp(\u), \lpp(\u'))}{\lpp(\u')}$ and $t= \frac{\lcm(\lpp(\u), \lpp(\u'))}{\lpp(\u)}$ such that $\lpp(t'\u') = \lpp(t\u)$.

\item $\lpp(t'f') = \lpp(tf)$ and  $(\u', f')$ is added to $G$ later than $(\u, f)$.

\end{enumerate}

We propose a new algorithm (named by NEW) based on the above criterion. This new algorithm  can be considered as an improved version of GVW. We have implemented F5, GVW and NEW on Singular (version 3-1-2) with the same structure, and no special optimizations (including matrical reduction) is used such that the timing is only affected by the effect of criteria. The timings were obtained  on Core i5 $4\times 2.8$ GHz with 4GB memory running Windows 7.

Another purpose of the comparison is to see the influences of computing orders of critical pairs. So we use two strategies for selecting critical pairs.  \smallskip\\
Minimal {\bf S}ignature Strategy: $[t_f(\u, f), t_g(\v, g)]$ is selected from $CPairs$ if there does {\em not} exist $[t_{f'}(\u', f'), t_{g'}(\v', g')]\in CPairs$ such that $\lpp(t_{f'} \u') \prec \lpp(t_f\u)$;\smallskip\\
Minimal {\bf D}egree Strategy: $[t_f(\u, f), t_g(\v, g)]$ is selected from $CPairs$ if there does {\em not} exist $[t_{f'}(\u', f'), t_{g'}(\v', g')]\in CPairs$ such that $\deg(\lpp(t_{f'} f')) \prec \deg(\lpp(t_f f))$.\smallskip \\
The proof in last section ensures the algorithms, including GVW,  are correct  using any of the above strategies.

In the following table, we use (s) and (d) to refer the two strategies respectively. The order $\prec_1$ is graded reverse lex order and $\prec_2$ is extended from $\prec_1$ in the following way: $x^\alpha\e_i \prec_2 x^\beta\e_j$, if either $\lpp(x^\alpha f_i) \prec_1 \lpp(x^\beta f_j)$, or  $\lpp(x^\alpha f_i) = \lpp(x^\beta f_j)$ and $i > j$. This order $\prec_2$ has also been used in \citep{Gao10b, SunWang10b}. The examples are selected from \citep{Gao10b}.

\begin{table}[!ht] \label{data}
\centering \caption{$\#All$: number of all critical pairs generated in the computation; $\#red$: number of critical pairs that are really reduced in the computation; $\#gen$: number of generators in the \gr basis in the last iteration but before computing a reduced \gr basis.}
\scriptsize
\begin{tabular}{|c|c|c|c|c|c|c|} \hline
 & F5(s) & gvw(s) & new(s) &  F5(d) & gvw(d) & new(d) \\ \hline\hline

 \multicolumn{7}{|c|}{Katsura5 (22 generators in reduced \gr basis)} \\ \hline

$\#all$ & 351 & 351& 351 & 378 & 351 & 378 \\ \hline

$\#red.$ & 39 & 39 & 39 & 40 & 39 & 40 \\ \hline

$\#gen.$ & 27 & 27 & 27 & 28 & 27 & 28 \\ \hline

time & 1.730 & 1.425 & 1.400 & 1.530 & 1.230 & 1.195 \\ \hline\hline

 \multicolumn{7}{|c|}{Katsura6 (41)} \\ \hline

$\#all$ & 1035 & 1035 & 1035 & 1225 & 1225 & 1275 \\ \hline

$\#red.$ & 73 & 73 & 73 & 77 & 77 & 78 \\ \hline

$\#gen.$ & 46 & 46 & 46 & 50 & 50 & 51\\ \hline

time & 10.040 & 8.715 & 7.865 & 7.520 & 6.920 & 5.650 \\ \hline\hline

 \multicolumn{7}{|c|}{Katsura7 (74)} \\ \hline

$\#all$ & 3240 & 3160 & 3160 & 3240 & 3240 & 3160 \\ \hline

$\#red.$ & 122 & 120 & 121 & 122 & 121 & 121 \\ \hline

$\#gen.$ & 81 & 80 & 80 & 81 & 81 & 80 \\ \hline

time & 47.840 & 70.371 & 38.750 & 39.440 & 74.535 & 29.950 \\ \hline\hline

 \multicolumn{7}{|c|}{Katsura8 (143)} \\ \hline

$\#all$ & 12880 & 11325 & 11325 & 12880 & 11476 & 11325 \\ \hline

$\#red.$ & 252 & 242 & 244 & 252 & 243 & 244 \\ \hline

$\#gen.$ & 161 & 151 & 151 & 161 & 152 & 151 \\ \hline

time & 426.402 & 2013.28 & 395.844 & 329.390 & 2349.16 & 310.908 \\ \hline\hline

 \multicolumn{7}{|c|}{Cyclic5 (20)} \\ \hline

$\#all$ & 1128 & 1128 & 1128 & 2211 & 1953 & 2080\\ \hline

$\#red.$ & 56 & 56 & 56 & 80 & 76 & 78 \\ \hline

$\#gen.$ & 48 & 48 & 48 & 67 & 63 & 65 \\ \hline

time & 3.074 & 2.953 & 2.708 & 2.864 & 2.654 & 2.630 \\ \hline\hline

 \multicolumn{7}{|c|}{Cyclic6 (45)} \\ \hline

$\#all$ & 19110 & 18528 & 18528 & 293761 & 81406 & 299925 \\ \hline

$\#red.$ & 234 & 231 & 231 & 821 & 463 & 834 \\ \hline

$\#gen.$ & 196 & 193 & 193 & 767 & 404 & 775 \\ \hline

time & 111.095 & 106.736 & 87.899 & 787.288 & 121.768 & 593.947 \\ \hline

\hline\end{tabular}
\end{table}

From the above table, we can see that the new algorithm   usually has better performance than the others. There are probably two main reasons. First, the new algorithm and GVW reject the same kind of critical pairs, but GVW's second criterion need to do some extra reductions before rejecting redundant critical pairs. Second, the critical pairs rejected by the new algorithm generally have larger leading power products than those rejected by F5 such that reductions in the new algorithm cost less time.

From the above table, we find that for some examples the algorithm with minimal
signature strategy has better performance. The possible reason is that less critical pairs are generated by this strategy.
For other examples,  the algorithm with minimal degree strategy cost less time.
The possible reason is that, although  the algorithm with the minimal degree strategy usually generates more critical pairs, the  critical pairs which are really needed to be reduced    usually have lower degrees.

\section{Conclusions and Future works}

A generalized criterion for signature related algorithms is proposed in this paper. We show in detail that this generalized criterion can specialize to the criteria of F5 and GVW by using appropriate admissible orders. Moreover, we also proved that if the partial order is admissible, the generalized criterion is always correct no matter which computing order of the critical pairs is used in the algorithm. Since the generalized criterion can specialize to the criteria of F5 and GVW, the proof in this paper also ensures the correctness of F5 and GVW  for any computing order of critical pairs.

The significance of this generalized criterion is to describe what kind of criterion is correct in signature related algorithms. The generalized criterion also provides a general approach to check and develop new criteria for signature related algorithms, i.e., if a new criterion can be specialized from the generalized criterion by using an admissible partial order, it must be correct; when developing new criteria, it suffices to choose admissible partial orders in the generalized criterion. We also develop a new criterion in this paper. We claim that if the admissible partial order is in fact a total order, then the generalized criterion can reject almost all useless critical pairs. The proof of the claim will be included in  future works.

On the algorithm GBGC, there are several open problems.

\noindent
{\bf Problem 1:} Is the generalized criterion still correct if the partial order is not admissible? We do know some partial order will lead to wrong criterion. For example, consider the following partial order which is not admissible: we say $(\u', f')<(\u, f)$, where $(\u, f), (\u', f')\in G$, if $f'=0$ and $f\not=0$; otherwise, $(\u', f')$ is added to $G$ {\em earlier} than $(\u, f)$. The above partial order leads to a wrong criterion. The reason is that $(\e_1, f_1), \cdots, (\e_m, f_m)$ are added to $G$ earlier than others, so using this partial order, the generalized criterion will reject almost all critical pairs generated later, which definitely leads to a wrong output unless $\{(\e_1, f_1), \cdots, (\e_m, f_m)\}$ itself is an S-\gr basis. Perhaps some partial orders lead to correct criteria, and this will be studied in the future.

\noindent
{\bf Problem 2:} Does the algorithm GBGC always terminate in finite steps?

\section{Acknowledgement}

We would like to thank Shuhong Gao and Mingsheng Wang for constructive discussions.

\end{document}